\documentclass[10pt,conference]{IEEEtran}
\IEEEoverridecommandlockouts

\newif\ifEditMode
\EditModefalse

\usepackage{graphicx, amsmath, amsfonts,amssymb, tikz, mathtools, tikz-cd,epstopdf,verbatim, hyperref}
\usepackage{wrapfig,xspace,subcaption}
\usepackage[font=itshape]{quoting}
\usepackage{enumitem}
\usepackage{amsbsy,scalerel}
\usepackage[T1]{fontenc}
\usepackage{algorithm}
\usepackage[noend]{algpseudocode}
\usepackage{syntax}
\usetikzlibrary{automata, positioning, arrows}

\usepackage{amsthm}
\usepackage[breakable, theorems, skins]{tcolorbox}
\usepackage{stmaryrd}
\usepackage{booktabs,multirow}

\tikzset{
->, 
node distance=3cm, 
every state/.style={thick, fill=gray!10}, 
initial text=$ $, 
}

\PassOptionsToPackage{bookmarks=false}{hyperref}

\newtheorem{example}{Example}[section]
\newtheorem{proposition}{Proposition}[section]
\newtheorem{definition}{Definition}[section]
\newtheorem{corollary}{Corollary}[section]
\newtheorem{lemma}{Lemma}[section]
\newtheorem{theorem}{Theorem}[section]

\graphicspath{{figures/}}

\allowdisplaybreaks

\ifEditMode
   \paperwidth=\dimexpr \paperwidth + 8cm\relax
   \oddsidemargin=\dimexpr\oddsidemargin + 4cm\relax
   \evensidemargin=\dimexpr\evensidemargin + 4cm\relax
   \marginparwidth=\dimexpr \marginparwidth + 4cm\relax
\fi

\DeclareRobustCommand{\mybox}[2][gray!20]{%
\begin{tcolorbox}[   
        breakable,
        left=0pt,
        right=0pt,
        top=0pt,
        bottom=0pt,
        colback=#1,
        colframe=#1,
        width=\dimexpr\columnwidth\relax, 
        enlarge left by=0mm,
        boxsep=5pt,
        arc=0pt,outer arc=0pt,
        ]
        #2
\end{tcolorbox}
}
\usepackage{xcolor}

\newcommand{\setunion}[0]{\boldsymbol{\cup}}

\newcommand{\setArg}[2]{\left\{ {#1} \,\,\middle|\,\, {#2} \right\}}

\renewcommand{\paragraph}[1]{\noindent\textbf{#1}}

\renewcommand{\tt}{\texttt}

\newcommand{\reals}[0]{\mathbb{R}}

\newcommand{\transp}[0]{\intercal}

\newcommand{\coeff}[0]{\textproc{coeff}}
\newcommand{\sign}[0]{\textproc{sign}}
\newcommand{\length}[0]{\textproc{length}}
\newcommand{\diag}[0]{\textproc{diag}}

\newcommand{\dom}[0]{\mathrm{dom}}

\begin{document}
\title{Context-Aided Variable Elimination for \\ Requirement Engineering
\thanks{This work was supported by the DARPA LOGiCS project under contract FA8750-20-C-0156 and by NSF and ASEE through an eFellows postdoctoral fellowship.}
}

\author{
\IEEEauthorblockN{Inigo Incer}
\IEEEauthorblockA{
\textit{CMS, Caltech}\\
Pasadena, CA, USA
}
\and
\IEEEauthorblockN{Albert Benveniste}
\IEEEauthorblockA{
\textit{INRIA/IRISA}\\
Rennes, France
}
\and
\IEEEauthorblockN{Richard M. Murray}
\IEEEauthorblockA{
\textit{CDS, Caltech}\\
Pasadena, CA, USA
}
\and
\IEEEauthorblockN{Alberto Sangiovanni-Vincentelli}
\IEEEauthorblockA{
\textit{EECS, UC Berkeley}\\
Berkeley, CA, USA
}
\and
\IEEEauthorblockN{Sanjit A. Seshia}
\IEEEauthorblockA{
\textit{EECS, UC Berkeley}\\
Berkeley, CA, USA
}
}

\maketitle
\begin{abstract}
Deriving system-level specifications from component specifications usually involves the elimination of variables that are not part of the interface of the top-level system.
This paper presents algorithms for eliminating variables from formulas by computing refinements or relaxations of these formulas in a context. We discuss a connection between this problem and optimization and give efficient algorithms to compute refinements and relaxations of linear inequality constraints.
\end{abstract}


\section{Introduction}

In the setting of requirement engineering using assume-guarantee specifications \cite{BenvenisteContractBook,Incer:EECS-2022-99,DBLP:journals/ejcon/Sangiovanni-VincentelliDP12}, we come across the need to eliminate variables from a formula by
computing refinements or relaxations in a context.
Let $\phi$ be a formula containing some variables that must be eliminated. These will be called \emph{irrelevant variables}, and the set of such variables will be denoted $Y$. In order to carry out the elimination, suppose we can use information from a set of formulas $\Gamma$ called the \emph{context}.
\mybox{
We will consider the problems of synthesizing missing formulas in the expressions
\begin{align*}
    \Gamma \land \text{ ? } \models \phi
    \quad \text{and} \quad
    \Gamma \land \phi \models \text { ?}
\end{align*}
such that the result lacks irrelevant variables.}
We will call the first problem \emph{antecedent synthesis}, and the second \emph{consequent synthesis}.
If $\psi$ is a solution to the antecedent synthesis problem, we will say that $\psi$ is a $Y$-antecedent (or a $Y$-refinement) of $\phi$ in the context $\Gamma$.
If $\psi$ is a solution to the consequent synthesis problem, we will say that $\psi$ is a $Y$-consequent (or a $Y$-relaxation) of $\phi$ in the context $\Gamma$.
This problem appears in requirement engineering in the following situations.

\begin{figure}[ht]
    \centering
    \includegraphics[width=0.8\columnwidth]{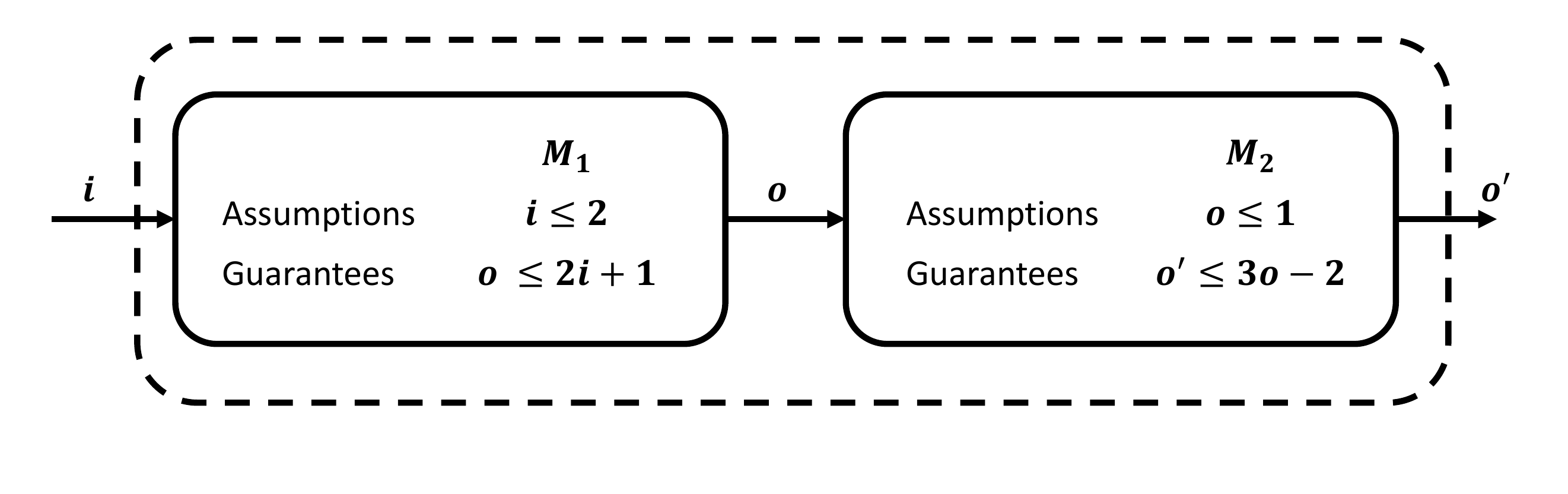}
    \caption{Two components connected in series. We wish to compute the specification of the top-level system composed of these two elements.}
    \label{fg:twobuffers}
\end{figure}

Figure \ref{fg:twobuffers} shows two components connected in series, $M_1$ and $M_2$. The first has input $i$ and output $o$, and the second has input $o$ and output $o'$. Each component comes with its assumptions and guarantees. The natures of $M_1$ and $M_2$ are left abstract; they could be routines executing in order, or they could be physical systems that interact through their input and output ports.
Our problem is to 
obtain a specification for the entire system 
using the specifications of the subsystems in such a way that only the top-level input and output variables $i$ and $o'$ appear in the final answer. In other words, the top-level specification should not mention the internal variable $o$.

We would like to operate the system in such a way that the assumptions of the two components hold. This would mean that we can rely on the two subsystems to deliver their guarantees. Thus, the top-level system should assume $(i \le 2) \land (o \le 1)$. This cannot be the top-level specification because the second formula involves the \emph{irrelevant variable} $o$. We would like to find a term only depending on $i$ that somehow ensures that the assumptions $o \le 1$ of $M_2$ are satisfied.
For this, we make use of the knowledge that $M_1$ guarantees $o \le 2i + 1$ when $i \le 2$.
We want to transform the constraint $o \le 1$ into a constraint $\psi$ on the input $i$ with the property that, given the guarantees of $M_1$, $\psi$ implies $o \le 1$. That is, this new constraint should satisfy $\psi \land (o \le 2i + 1) \rightarrow (o \le 1)$, which means that $\psi$ should be a refinement (an antecedent) of $o \le 1$ in the context of the guarantees of $M_1$. We observe that
$\psi\colon i \le 0$ satisfies this requirement. Thus, we transform the term $o \le 1$ into the term $i \le 0$. The top-level assumptions become $i \le 0$. We can verify that these top-level assumptions ensure that subsystems $M_1$ and $M_2$ have their assumptions met.

\newcommand{\qeda}{\hfill\rule{1ex}{1ex}}

Similarly, the guarantees for the system are $(o \le 2i + 1) \land (o' \le 3 o - 2)$. Again, the variable $o$ is not welcome in the final answer, giving us two options: we could eliminate both terms and have no guarantees---which is right, but not useful---or we could relax (compute the consequent of) one of the terms in the context of the other term. We find out, for example, that $(o \le 2i + 1) \land (o' \le 3 o - 2) \rightarrow (o' \le 6i + 1)$. The constraint $o' \le 6i + 1$ is an acceptable promise for the system specification.

By computing antecedents and consequents, we concluded that the top-level system guarantees $o' \le 6i + 1$ as long as the input satisfies $i \le 0$.

\medskip

\begin{figure}[t]
    \centering
    \includegraphics[width=0.8\columnwidth]{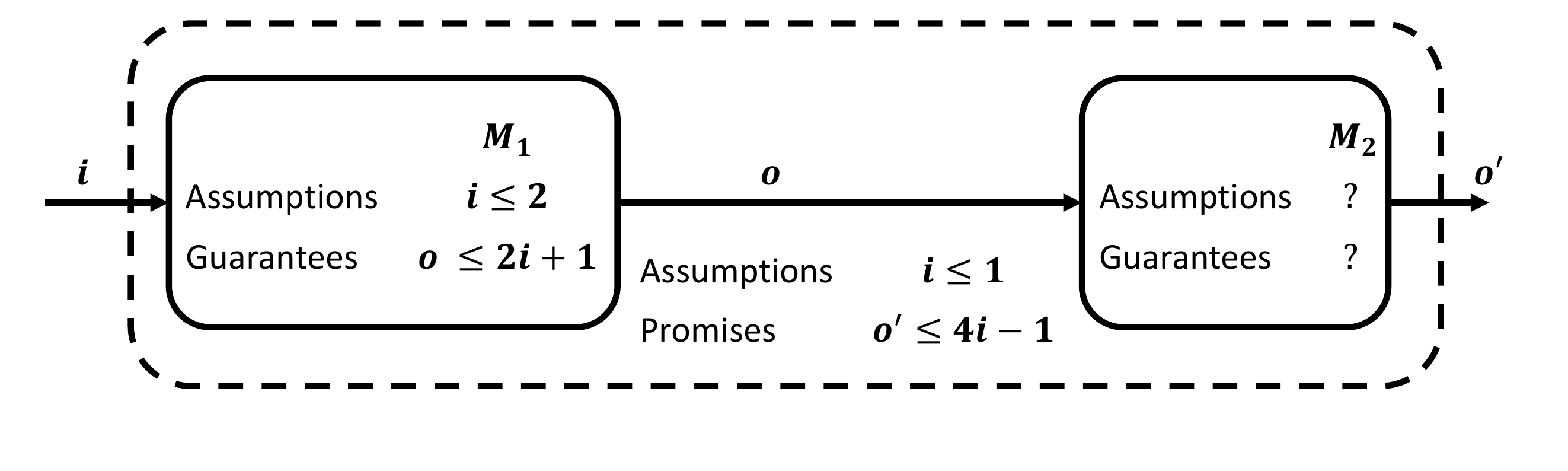}
    \caption[Quotient]{Two components connected in series. We are given the specification of the top-level system and the specification of $M_1$. The problem is to find the pre/post conditions of $M_2$ in order to obtain the given system-level specifications.}
    \label{fg:twobuffersquo}
\end{figure}

This example shows that the computation of antecedents and consequents plays a key role in the identification of assume-guarantee pairs or pre/post conditions. One may be tempted to link antecedents to assumptions and consequents to guarantees. This is not always so. Figure \ref{fg:twobuffersquo} show a situation in which we again have two components connected in series, $M_1$ and $M_2$, with inputs and outputs as before. Now we are given the top level assumptions and guarantees, and we also know the assumptions and guarantees of $M_1$. The problem is to find the pre/post conditions of $M_2$ using this data.

We know that the top level assumes that $i \le 1$. Under these assumptions, $M_1$ guarantees $o \le 2i + 1$.
The assumptions of $M_2$ should be met when the top-level system is operating within its assumptions. Thus, the assumptions of $M_2$ should be implied by the data $(i \le 1) \land (o \le 2i + 1)$. Since the assumptions of $M_2$ should only depend on $o$, we obtain the expression $o \le 3$.

Now we look for the guarantees of $M_2$, which we call $\psi$. The guarantees of $M_1$ and $M_2$ together must imply the top-level guarantees. Thus, we have the expression $\psi \land (o \le 2i + 1) \rightarrow (o' \le 4i - 1)$. In other words, $\psi$ is an antecedent of $o \le 2i + 1$ in the context $o' \le 4i - 1$. We require $\psi$ to only refer to variables $o$ and $o'$ and observe that $o' \le 2o - 3$ is an acceptable promise.

We conclude that $M_2$ should assume $o \le 3$ and promise $o' \le 2o - 3$.

\medskip

The examples just described motivate us to study automated mechanisms for the computation of antecedents and consequents of formulas in a given context with the objective of removing dependencies on irrelevant variables.
We first consider this problem for general first-order formulas and then specialize to the situation when
formulas are expressed as linear constraints in a context of linear inequalities. We provide efficient algorithms to address this problem.
Our previous discussion shows that this problem is of relevance to requirement engineering.



\section{Computing antecedents and consequents in first-order logic}
\label{sc:pord}

Suppose $\phi$ is a formula in first-order logic with free variables $x = (x_1, \ldots, x_m)$ and $y = (y_1, \ldots, y_n)$, and
$\Gamma$ a formula with free variables $x$, $y$, and $z = (z_1, \ldots, z_o)$.
As a matter of notation, when the free variables of a formula are understood, we will simply write the name of the formula, e.g., $\phi(x,y)$ is synonymous with $\phi$.
We let $Y$ be the set of \emph{irrelevant variables} that we want to eliminate from $\phi$. Throughout this paper, the set of irrelevant variables is always $Y = \{y_i\}_i$, i.e., we are always interested in eliminating the $y$ variables from $\phi$.
\begin{definition}\label{bjasjbcgw}
We say that a formula $\psi(x,z)$ with free variables $x = (x_1, \ldots, x_m)$ and $z=(z_1, \ldots, z_o)$ is a $Y$-antecedent of $\phi$ in the context $\Gamma$ if
$$
\Gamma(x,y,z) \land \psi(x,z) \models \phi(x,y).
$$
We say that $\psi(x,z)$ is a $Y$-consequent of $\phi$ in the context $\Gamma$ if
$$
\Gamma(x,y,z) \land \phi(x,y) \models \psi(x,z).
$$
\end{definition}

\mybox{
Our objective is to eliminate the irrelevant variables $y$ from $\phi$ by synthesizing a $Y$-antecedent or $Y$-consequent of $\phi$ in the context $\Gamma$.}

This section contains the following results:
\begin{enumerate}
\item a characterization of the optimal solutions for $Y$-antecedent/consequent synthesis in general (Proposition~\ref{kugbkgk});
\item a characterization of the optimal solution to this problem when the context can be expressed as a conjunction of a formula that depends on irrelevant variables and a formula that does not (Proposition~\ref{kjbhgbjk});
\item a characterization of optimal solutions to this problem when $\phi$ monotonically depends on a function of $y$ and does not depend on $y$ in any other way (Proposition~\ref{kgbjsfgbj}); and
\item methods to synthesize $Y$-antecedents/consequents compositionally from the syntax of $\phi$ and $\Gamma$ (Propositions~\ref{kbqgkqbg} and~\ref{kqxnfkxr}).
\end{enumerate}

We define two formulas obtained from $\phi$, $\Gamma$, and $Y$. We then discuss an important property of these formulas.

\begin{definition}
Given $\phi$, $\Gamma$, and $Y$ as above, let
$$\phi_\Gamma \coloneqq \forall y.\; \left( \Gamma \rightarrow \phi \right)$$ and
$$\phi^\Gamma \coloneqq \exists y.\; \left( \Gamma \land \phi \right).$$
\end{definition}


\begin{proposition}\label{kugbkgk}
    A formula $\psi(x,z)$ is a $Y$-antecedent for $\phi$ in the context $\Gamma$ if and only if 
    $\psi \models \phi_\Gamma$.
    
    A formula $\rho(x,z)$ is a $Y$-consequent for $\phi$ in the context $\Gamma$ if and only if 
    $\phi^\Gamma \models \rho$.
\end{proposition}
\begin{proof}
    We have
    $\Gamma \land \psi \models \phi 
    \Leftrightarrow \psi \models (\Gamma \to \phi)
    \Leftrightarrow \psi \models \phi_\Gamma$.
    Similarly,
    $\Gamma \land \phi \models \rho \Leftrightarrow \exists y.\; (\Gamma \land \phi) \models \rho$.
\end{proof}

This result means that $\phi_\Gamma$ is the weakest $Y$-antecedent of $\phi$ in the context $\Gamma$. Conversely, $\phi^\Gamma$ is the strongest $Y$-consequent.

\begin{lemma}
The denotations of $\phi_\Gamma$ and $\phi^\Gamma$ are
$$\llbracket \phi_\Gamma \rrbracket = \bigcap_{\substack{b \in \dom(y)\\ \Gamma(x,b,z)}} \llbracket \phi (x,b,z) \rrbracket \text{ and }
\llbracket \phi^\Gamma \rrbracket = \bigcup_{\substack{b \in \dom(y)\\ \Gamma(x,b,z)}} \llbracket \phi (x,b,z) \rrbracket.$$
\end{lemma}
\begin{proof}
    We compute
    \begin{align*}
    \llbracket \phi_\Gamma \rrbracket &=
    \llbracket \forall y\; (\Gamma \to \phi) \rrbracket =
    \bigcap_{b \in \dom(y)} \llbracket \Gamma(x,b,z) \to \phi(x,b) \rrbracket \\ &=
    \bigcap_{\substack{b \in \dom(y) \\ \Gamma(x,b,z)}} \llbracket \phi(x,b) \rrbracket.
    \end{align*}
    A similar reasoning applies to $\phi^\Gamma$.
\end{proof}

Proposition~\ref{kugbkgk} provides a universal characterization of $Y$-antecedents and consequents. 
We now state a result that allows us to synthesize these objects using partial information from the context.
\begin{proposition}\label{kjbhgbjk}
Suppose that the context $\Gamma$ can be written as $\Gamma(x,y,z) = \Gamma_1(x,z) \land \Gamma_2(x,y,z)$. Then
$\phi_{\Gamma} = \Gamma_1 \to \phi_{\Gamma_2}$ and
$\phi^{\Gamma} = \Gamma_1 \land \phi_{\Gamma_2}$.
\end{proposition}
\begin{proof}
From Definition~\ref{bjasjbcgw}, we have
\begin{align*}
\phi_\Gamma &=
\forall y.\; ((\Gamma_1 \land \Gamma_2) \to \phi) \\ &=
\forall y.\; (\Gamma_1(x,z) \to (\Gamma_2(x,y,z) \to \phi(x,y))) \\ &=
\Gamma_1 \to ( \forall y.\; (\Gamma_2(x,y,z) \to \phi(x,y))) =
\Gamma_1 \to \phi_{\Gamma_2}
\intertext{and}
\phi_\Gamma &=
\exists y.\; ((\Gamma_1(x,z) \land \Gamma_2(x,y,z)) \land \phi(x,y)) \\ &=
\Gamma_1 \land (\exists y.\; (\Gamma_2(x,y,z)) \land \phi(x,y)) =
\Gamma_1 \land \phi^{\Gamma_2}.\qedhere
\end{align*}
\end{proof}

An immediately corollary of Proposition~\ref{kjbhgbjk} is the fact that
$\phi_\Gamma \land \Gamma = \phi_{\Gamma_2} \land \Gamma$ and
$\phi^\Gamma \land \Gamma = \phi^{\Gamma_2} \land \Gamma$. This means that $\phi_{\Gamma_2}$ and $\phi^{\Gamma_2}$ are optimal in the context $\Gamma$ and thus can be used as the optimal $Y$-antecedents and consequents, respectively. The following example shows that quantifying over smaller contexts can yield preferable results.

\begin{example}
Suppose that $\phi$ and $\Gamma$ are formulas in linear arithmetic
such that $\phi(x,y) = (x \le y)$ and $\Gamma(x,y,z) = \Gamma_1(x,z) \land \Gamma_2(x,y,z)$, where $\Gamma_1 = (z \le 2)$ and $\Gamma_2 = (z \le y)$. From Proposition \ref{kugbkgk}, the optimal $Y$-antecedent of $\phi$ in $\Gamma$ is
\begin{align*}
\phi_\Gamma &= \forall y. \; ((z \le 2) \land (z \le y) \to (x \le y)) \\ &= ((z \le 2) \to (x \le z)).
\end{align*}
As we just discussed, we may as well use $\phi_{\Gamma_2}$ as the $Y$-antecedent of $\phi$. This formula is
\begin{align*}
\phi_{\Gamma_2} &= \forall y. \; ((z \le y) \to (x \le y)) = (x \le z).
\end{align*}
For requirement engineering, it is preferable to output $\phi_{\Gamma_2}$ instead of $\phi_{\Gamma}$ because it is syntactically simpler.
\qeda
\end{example}

Now we discuss an optimization formulation of antecedent/consequent synthesis.

\begin{proposition}\label{kgbjsfgbj}
Suppose that $\phi$ can be expressed as $\phi(x,g(y))$, where $\phi$ is monotonic in the second argument. Define
{\small
\begin{align*}
g^-(a,c) =
\begin{cases}
\begin{aligned}[t]
& \underset{b \in \reals^n}{\text{minimize}}
& & g(b) \\
& \text{subject to}
& & [x\coloneqq a, y \coloneqq b, z \coloneqq c] \models \Gamma
\end{aligned}
\end{cases}
\end{align*}}
and
{\small
\begin{align*}
g^+(a,c) =
\begin{cases}
\begin{aligned}[t]
& \underset{b \in \mathbb{R}^n}{\text{maximize}}
& & g(b) \\
& \text{subject to}
& & [x\coloneqq a, y \coloneqq b, z \coloneqq c] \models \Gamma.
\end{aligned}
\end{cases}
\end{align*}
}
Then we have $\phi(x, g^-(x,z)) \models \phi_\Gamma$ and
$\Gamma \land \phi_\Gamma \models \phi(x, g^-(x,z))$. Similarly,
$\phi^\Gamma \models \phi(x, g^+(x,z))$ and
$\phi(x, g^+(z,z) \land \Gamma \models \phi^\Gamma$.
\end{proposition}
\begin{proof}
If $[x\coloneqq a, y\coloneqq b, z\coloneqq c] \not \models \Gamma$ for all $b \in \dom(y)$, then 
$[x\coloneqq a, z\coloneqq c] \models \phi_\Gamma$. Otherwise,
$\bigcap_{\substack{b \in \dom(y) \\ \Gamma(a,b,c)}} \llbracket \phi(a, g(b)) \rrbracket = \llbracket \phi(a, g^-(a,c)) \rrbracket$.
The second part is proved similarly.
\end{proof}

From Propositions~\ref{kugbkgk} and \ref{kgbjsfgbj}, we know that
$\phi(x, g^-(x,z))$ and $\phi(x, g^+(x,z))$ are, respectively, a $Y$-antecedent and a $Y$-consequent of $\phi$ in the context $\Gamma$.
Moreover, $\phi(x, g^-(x,z))$ and $\phi(x, g^+(x,z))$ have the same denotations as the optimal solutions
$\phi_\Gamma$ and $\phi^\Gamma$, respectively, \emph{in the context $\Gamma$}.
Thus, we regard $\phi(x, g^-(x,z))$ and $\phi(x, g^+(x,z))$ as optimal.

\begin{example}\label{bghgh}
Suppose we want to compute an antecedent of the formula
$$\phi\colon (p \land q) \lor r$$
in the context $\Gamma\colon s \rightarrow q$, where $q$ is an irrelevant variable.
Let $g(q) = q$ and $\phi(x, g(q)) = (p \land q) \lor r$. Then $\phi$ is monotonic in its last argument. To apply Proposition~\ref{kgbjsfgbj}, we compute $g^-$:
\begin{align*}
    g^-(a,c) =&
    \begin{cases}
    \begin{aligned}[t]
    & \underset{b \in \{0,1\}}{\text{minimize}}
    & & q \\
    & \text{subject to}
    & & [p\coloneqq a, q \coloneqq b, s\coloneqq c] \models s \to q
    \end{aligned}
    \end{cases} \\
=&
c.
\end{align*}
By Proposition~\ref{kgbjsfgbj}, we conclude that $\phi(p, g^-(p,s))$ is a $Y$-antecedent of $\phi$ in the given context, i.e., we get the antecedent
$$
(p \land s) \lor r.
$$
In contrast, we compute $\phi_\Gamma \colon \forall y. \; ((s \rightarrow q) \rightarrow ((p\land q) \lor r)) = (p \lor r) \land ( s \lor r)$.
\qeda
\end{example}

In Example~\ref{bghgh}, $\phi_\Gamma$ and $\phi(x, g^-(x,z))$ match. Yet, observe that the latter immediately yields a result in a syntactic form which is closer to the original $\phi$. This happens because this expression is obtained by simply replacing $q$ with $g^-(x,z)$ in $\phi$.
Having results which are syntactically similar to the original expressions is important in requirement engineering, as the syntax of requirements entered by users has a close connection to the semantics they want to express.

\subsection{Compositional results}

We now express $\phi$ and $\Gamma$ as
$\phi = \bigvee_j \bigwedge_i \phi_i^j(x,y)$, where the $\phi_i^j(x,y)$ are clauses (i.e., terms formed from atomic formulas and Boolean connectives), and
$\Gamma = \bigvee_k \Gamma^k(x,y,z)$, where each $\Gamma^k(x,y,z)$ is a conjunction of clauses.
Instead of synthesizing an antecedent or consequent for $\phi$ directly, we will seek methods to do this compositionally.

\begin{proposition}\label{kbqgkqbg}
    Let $\tilde \phi_\Gamma \coloneqq \bigvee_j \bigwedge_{i,k} \forall y.\; \left( \Gamma^k \rightarrow \phi_i^j \right)$
    and    $\tilde \phi^\Gamma \coloneqq \bigvee_{k,j} \bigwedge_{i} \exists y.\; \left( \Gamma^k \land \phi_i^j \right)$. 
    $\tilde \phi_\Gamma$ is an antecedent and $\tilde \phi^\Gamma$ a consequent of $\phi$ in the context $\Gamma$.
\end{proposition}
\begin{proof}
$
\tilde \phi_\Gamma =
\bigvee_j \bigwedge_{i,k} \forall y.\; ( \Gamma^k \rightarrow \phi_i^j ) =
\bigvee_j \forall y.\; ( \Gamma \rightarrow \bigwedge_i \phi_i^j ) \models
\forall y.\; \bigvee_j ( \Gamma \rightarrow \bigwedge_i \phi_i^j ) = \phi_\Gamma
$.
By applying Proposition~\ref{kugbkgk}, we proved the first part. We also have
$
\phi^\Gamma =
\exists y.\; ( \Gamma \land \phi ) =
\exists y.\; \bigvee_{k,j} \bigwedge_{i} ( \Gamma^k \land \phi_i^j ) =
\bigvee_{k,j} \exists y.\; \bigwedge_{i} ( \Gamma^k \land \phi_i^j ) \models
\bigvee_{k,j} \bigwedge_{i} \exists y.\; ( \Gamma^k \land \phi_i^j ) =
\tilde \phi^\Gamma
$,
which shows the second part after applying Proposition~\ref{kugbkgk}.
\end{proof}

The compositional result of Proposition~\ref{kbqgkqbg} allows us to focus on the case in which $\phi$ is a clause and $\Gamma$ is a conjunction of clauses. \emph{We shall assume this from now on}.

To further exploit compositionality, express $\phi$ as
\begin{equation}\label{jhgfkjgb}
\phi(x,y) = \phi(x, g_1(y), \ldots, g_p(y)),
\end{equation}
where $\phi$ is required to be monotonic in all arguments, except the first. We assume that our language allows any clause to be expressed in this way. This is true for linear arithmetic and for real arithmetic with a ReLu (rectified linear unit) function.

\begin{example}\label{kajhcghsk}
Suppose $\phi$ is the real-arithmetic formula $x^2 - x y + y^2 \le 0$. This is equivalent to $x^2 - r(x) y + r(-x) y + y^2 \le 0$, where $r$ is the ReLu function. We can thus write $\phi$ as $\phi(x, g_1(y), g_2(y), g_3(y))$ with $g_1(y) = y$, $g_2(y) = -y$, and $g_3(y) = -y^2$.
\qeda
\end{example}

We now study how to exploit the structure \eqref{jhgfkjgb} to compute $Y$-antecedents and consequents for $\phi$ in $\Gamma$.

\begin{proposition}\label{kqxnfkxr}
Let
{\small
\begin{align*}
g_i^-(a,c) =
\begin{cases}
\begin{aligned}[t]
& \underset{b \in \reals^n}{\text{minimize}}
& & g_i(b) \\
& \text{subject to}
& & [x\coloneqq a, y \coloneqq b, z \coloneqq c] \models \Gamma
\end{aligned}
\end{cases}
\end{align*}}
and
{\small
\begin{align*}
g_i^+(a,c) =
\begin{cases}
\begin{aligned}[t]
& \underset{b \in \mathbb{R}^n}{\text{maximize}}
& & g_i(b) \\
& \text{subject to}
& & [x\coloneqq a, y \coloneqq b, z \coloneqq c] \models \Gamma.
\end{aligned}
\end{cases}
\end{align*}
}
The formulas
$\phi(x, g_1^-(x,z), \ldots, g_p^-(x,z))$ and
$\phi(x, g_1^+(x,z), \ldots, g_p^+(x,z))$ are, respectively,
$y$-antecedents and consequents of $\phi$ in $\Gamma$.
\end{proposition}
\begin{proof}
From Proposition \ref{kugbkgk}, we have
\begin{align*}
&\llbracket \phi_\Gamma \rrbracket =
\bigcap_{\substack{b \in \dom(y) \\ \Gamma(x,b,z)}} \llbracket \phi \rrbracket
\supseteq
\llbracket \phi(x, g_1^-(x,z), \ldots, g_p^-(x,z)) \rrbracket.
\end{align*}
We conclude that
$\phi(x, g_1^-(x,z), \ldots, g_p^-(x,z)) \models \phi_\Gamma$. Proposition~\ref{kugbkgk} yields the first part. The second part is proved similarly.
\end{proof}

Proposition~\ref{kqxnfkxr} tells us that we can independently optimize over the monotonic functions composing $\phi$ to compute $Y$-antecedents and consequents. This result does not yield the optimality guarantees of Proposition~\ref{kgbjsfgbj}, but it allows us to compute antecedents/consequents compositionally.

\begin{example}
Continuing Example \ref{kajhcghsk}, suppose we want to compute a consequent of $\phi$ in the context $\Gamma \colon (y \le 2) \land (1 \le y)$.
We apply Proposition~\ref{kqxnfkxr} by optimizing over the $g_i$ separately. We obtain $g_1^+(x) = 2$, $g_2^+(x) = -1$, and $g_3^+(x) = -1$. The formula $\phi(x,g_1^+(x),g_2^+(x),g_3^+(x)) = x^2 - 2r(x) - r(-x) -1 \le 0$ is a $Y$-consequent of $\phi$ in $\Gamma$.
\end{example}




\section{Linear inequality constraints}

We apply the results of Section \ref{sc:pord} to the situation when atoms are linear inequalities, or polyhedral constraints. Polyhedral constraints are an intuitive formalism for writing specifications for complex systems, as they allow us to place piecewise linear bounds on quantities of interest.

We consider formulas of the form $\bigvee_j \bigwedge_i \phi_i^j(x,y)$, where $\phi_i^j(x,y)$ are atoms,
and contexts of the form
$\bigvee_k \Gamma^k(x,y,z)$, where $\Gamma^k(x,y,z)$ are conjunctions of atoms. 

Due to Proposition~\ref{kbqgkqbg}, we will focus our attention on algorithms for the efficient computation of antecedents and consequents when $\phi$ is an atom and $\Gamma$ a conjunction of atoms. $\phi$ will have the form
\begin{equation*}
    \phi\colon \quad \sum_{i = 1}^m p_i x_i + \sum_{i = 1}^n q_i y_i + r \le 0,
\end{equation*}
where $r$ and the $p_i$ and $q_i$ are constants. The set of irrelevant variables to be eliminated is $Y = \{y_i\}_i$.
The context $\Gamma$ is a set of linear inequalities of the form
\begin{equation*}
    \Gamma = \left \{ \sum_{j = 1}^m \alpha_{ij} x_i + \sum_{j = 1}^n \beta_{ij} y_i + \sum_{j = 1}^o \gamma_{ij} z_i  + K_i \le 0 \right \}_{i = 1}^N,
\end{equation*}
where the $K_j$, $\alpha_i^j$, $\beta_i^j$, and $\gamma_i^j$ are constants.


Let $A = (\alpha_{ij}) \in \reals^{N \times m}$, $B = (\beta_{ij}) \in \reals^{N \times n}$, $C = (\gamma_{ij}) \in \reals^{N \times o}$, $K \in \reals^N$, $p \in \reals^m$, and $q \in \reals^n$. We also let $x = (x_i)$, $y = (y_i)$, $z = (z_i)$ be $m$-, $n$-, and $o$-dimensional vectors of variables, respectively.

\mybox{Our problem is to eliminate the $y$ variables from
\begin{equation}
\label{jbhgbkgh}
\phi\colon \quad p^\transp x + q^\transp y + r \le 0
\end{equation}
using the context
\begin{equation}
\label{jvbhgkvgb}
\Gamma\colon \quad A x + By + Cz + K \le 0
\end{equation}
by computing $Y$-antecedents/consequents.
}

\medskip

Let $b(x, z) = -K - A x - C z$. We obtain the following corollary from Proposition~\ref{kgbjsfgbj}.

\begin{corollary}\label{nlkflkg}
Let $\phi$ and $\Gamma$ be as above. Let
\begin{align}\label{lxkasjf}
    g^-(x,z) =
    \begin{cases}
    \begin{aligned}[t]
    & \underset{y \in \reals^n}{\text{minimize}}
    & & -q^\transp y \\
    & \text{subject to}
    & & By \le b(x, z)
    \end{aligned}
    \end{cases}
\end{align}
and
\begin{align}\label{kdsgckg}
    g^+(x,z) =
    \begin{cases}
    \begin{aligned}[t]
    & \underset{y \in \reals^n}{\text{maximize}}
    & & -q^\transp y \\
    & \text{subject to}
    & & By \le b(x, z).
    \end{aligned}
    \end{cases}
\end{align}
Then the formula $p^\transp x - g^-(x,z) \le r$ is an optimal $Y$-antecedent of $\phi$ in the context $\Gamma$ and $p^\transp x - g^+(x,z) \le r$ is an optimal $Y$-consequent of $\phi$ in the context $\Gamma$.
\end{corollary}

\begin{example}
Suppose we wish to eliminate variables $y_1$ and $y_2$ from $2 x + y_1 - 2y_2 \le 5$ through antecedent computation, using the context $\{x - 2 y_1 + y_2 + z \le 1, 3 y_1 - 4 y_2 \le 6\}$. We compute
\begin{align*}
    g^-(x,z) &=
    \begin{cases}
    \begin{aligned}[t]
    & \underset{y_1, y_2 \in \reals}{\text{minimize}}
    & & -(y_1 - 2y_2) \\
    & \text{subject to}
    & & x - 2 y_1 + y_2 + z \le 1 \\
    & & & 3 y_1 - 4 y_2 \le 6
    \end{aligned}
    \end{cases} \\ &
    =
    -\left(4 - \frac{2}{5} (x + z)\right).
\end{align*}
The antecedent formula is $2x + 4 - \frac{2}{5}(x + z) \le 5$, which becomes $8x - 2z \le 5$.
\qeda
\end{example}

\begin{example}
Suppose we wish to eliminate variables $y_1$ and $y_2$ from $x + 5 y_1 - 2y_2 \le 5$ by the computation of a consequent, using the context $\{x - 2 y_1 + y_2 + z \le 1, 3 y_1 - 4 y_2 \le 6\}$. We compute
\begin{align*}
    g^+(x,z) &=
    \begin{cases}
    \begin{aligned}[t]
    & \underset{y_1, y_2 \in \reals}{\text{maximize}}
    & & -(5 y_1 - 2y_2) \\
    & \text{subject to}
    & & x - 2 y_1 + y_2 + z \le 1 \\
    & & & 3 y_1 - 4 y_2 \le 6
    \end{aligned}
    \end{cases} \\ &
    =\;\;
    -\left(-4 + \frac{14}{5}(x + z)\right).
\end{align*}
The consequent is $x -4 + \frac{14}{5}(x + z) \le 5$, or $19x + 14z \le 45$.
\qeda
\end{example}

\subsection{Solving the symbolic optimization problems}

Corollary~\ref{nlkflkg} provides an explicit expression to compute optimal $Y$-antecedents and consequents. 
The next issue we face is the computation of \eqref{lxkasjf} and \eqref{kdsgckg}. Both are linear programs, but their solutions are symbolic due to the presence of $b(x, z)$.
We observe that if we have a context $\Gamma'$ such that $\Gamma = \Gamma' \land \Gamma''$, and if $\phi'$ is a $Y$-antecedent/consequent of $\phi$ in the context $\Gamma'$ then $\phi'$ is also a $Y$-antecedent/consequent in the context $\Gamma$.
First, we will discuss conditions required for solving linear programs with symbolic constraints when the context has as many constraints as optimization variables (i.e., when $N = n$).
Then we will discuss approaches for selecting from $\Gamma$ a set of formulas that meets these requirements.
We consider two selection criteria: 
a method based on positive solutions to linear equations and a method based on linear programming.

\subsubsection{Optimization in a subset of the context}
A linear program achieves its optimal value on the boundary of its constraints. If the context $\Gamma$ contains $N$ constraints and is a bounded polyhedron, then the optimal value of the linear program will occur at one of the $\binom{N}{n}$ possible vertices.
We will look for ways to choose $n$ constraints from $\Gamma$ such that the optimization problems achieve optimal values at the vertex determined by those $n$ constraints. First, we focus on solving symbolic LPs when the context contains $n$ constraints. The following definition will be useful:

\begin{definition}
Let $M \in \reals^{n \times n}$ and $\nu \in \reals^{n}$. We say that $(M, \nu)$ is a refining pair if $M$ is invertible and $(M^\transp)^{-1} \nu$ has nonnegative entries.
We say that the pair $(M, \nu)$ is a relaxing pair if $M$ is invertible and $-(M^\transp)^{-1} \nu$ has nonnegative entries.
\end{definition}

As the next result shows, these conditions are sufficient to solve the problems \eqref{lxkasjf} and \eqref{kdsgckg} when there are as many context formulas as irrelevant variables (i.e., when $N = n$). Suppose $J \subseteq \{1, \ldots, N\}$ has cardinality $n$. We let $B_J = (\beta_{J_i,j})_{i,j = 1}^n$ and $b_J = (b_{J_i})_{i = 1}^n$ be the $J$-indexed rows of $B$ and $b$, respectively.

\begin{lemma}\label{lnqfxq}
Suppose $(B_J, q)$ is a refining pair. Then
\[\begin{cases}
\begin{aligned}[t]
& \underset{y \in \reals^n}{\text{maximize}}
& & q^\transp y \\
& \text{subject to}
& & B_Jy \le b_J(x, z)
\end{aligned}
\end{cases} = q^\transp B_J^{-1} b(x,z).\]
Suppose $(B_J, q)$ is a relaxing pair. Then
\[\begin{cases}
\begin{aligned}[t]
& \underset{y \in \reals^n}{\text{minimize}}
& & q^\transp y \\
& \text{subject to}
& & B_Jy \le b_J(x, z)
\end{aligned}
\end{cases} = q^\transp B_J^{-1} b_J(x,z).\]
\end{lemma}
\begin{proof}
Let $(B_J, q)$ be a refining pair. We consider the first problem and its Lagrange dual (see \cite{boyd2004convex}, Section 5.2.1):
\begin{align*}
\text{primal}&
\begin{cases}
\begin{aligned}[t]
    & \underset{y}{\text{minimize}}
    & & - q^\transp y \\
    & \text{subject to}
    & &  B_J y \le b_J(x,z)
\end{aligned}
\end{cases}
\\
\text{dual}&
\begin{cases}
\begin{aligned}[t]
    & \underset{\lambda}{\text{maximize}}
    & & -b_J^\transp \lambda \\
    & \text{subject to}
    & &  B_J^\transp \lambda - q= 0 \\
    & & & \lambda \ge 0
\end{aligned}
\end{cases}
\end{align*}
The dual problem only admits the solution $\lambda^\star = ( B_J^\transp)^{-1} q$ if $\lambda^\star \ge 0$, which is the case, as $(B_J, q)$ is a refining pair. Thus, the optimal value of the dual problem is $v^\star = -  q^\transp (  B_J^{-1} b_J) $. As strong duality holds for any linear program (see \cite{boyd2004convex}, Section 5.2.4), $v^\star$ is also the optimal value of the primal problem. The statement of the theorem follows.

Now suppose $(B_J, q)$ is a relaxing pair. We consider the second problem and its dual:
\begin{align*}
\text{primal}&
\begin{cases}
\begin{aligned}[t]
    & \underset{y}{\text{minimize}}
    & & q^\transp y \\
    & \text{subject to}
    & &  B_J y \le b_J(x,z)
\end{aligned}
\end{cases}
\\
\text{dual}&
\begin{cases}
\begin{aligned}[t]
    & \underset{\lambda}{\text{maximize}}
    & & -b_J^\transp \lambda \\
    & \text{subject to}
    & &  B_J^\transp \lambda + q= 0 \\
    & & & \lambda \ge 0
\end{aligned}
\end{cases}
\end{align*}
The dual only admits the solution $\lambda^\star = -(B_J^\transp)^{-1} q$ if $\lambda^\star \ge 0$, which is the case because $(B_J, q)$ is a relaxing pair. The optimal value of the dual problem is $v^\star =  q^\transp ( B_J^{-1} b)$. Due to strong duality, $v^\star$ is also the optimal value of the primal problem.
\end{proof}

As a consequence of Corollary \ref{nlkflkg} and Lemma \ref{lnqfxq}, we obtain the following result:
\begin{corollary}\label{lnhxql}
With all definitions as above, 
if $(B_J, q)$ is a refining pair, then
$p^\transp x + q^\transp B_J^{-1} b_J(x,z) \le r$ is a $Y$-antecedent of $p^\transp x + q^\transp y \le r$ in the context $B y \le b(x,z)$.
If $(B_J, q)$ is a relaxing pair, then
$p^\transp x + q^\transp B_J^{-1} b_J(x,z) \le r$ is a $Y$-consequent of $p^\transp x + q^\transp y \le r$ in the context $B y \le b(x,z)$.
\end{corollary}
Corollary \ref{lnhxql} gives explicit formulas for computing $Y$-antecedents/consequents of a formula in a context. This result is missing methods for computing $J$, the set of the indices of formulas in $\Gamma$, in such a way that it yields refining or relaxing pairs $(B_J, q)$, as needed. We consider two methods to identify $J$.

\subsubsection{Computing J by seeking positive solutions to linear equations}
Our first method is based on identifying constraints yielding linear systems of equations whose solutions are guaranteed to be nonnegative. We will use the following result.

\begin{theorem}[Kaykobad \cite{KAYKOBAD1985133}]\label{lnhdnlqwh}
Let $M = (\mu_{ij}) \in \reals^{n \times n}$ and $\nu \in \reals^n$. Suppose the entries of $M$ are nonnegative, its diagonal entries are positive, the entries of $\nu$ are positive, and $\nu_i > \sum_{\substack{j \ne i}} \mu_{ij} \frac{\nu_j}{\mu_{jj}}$ for all $i \le n$. Then $M$ is invertible and $M^{-1} \nu$ has positive entries.
\end{theorem}

\begin{definition}
A pair $(M, \nu)$, where $M = (\mu_{ij}) \in \reals^{n \times n}$ and $\nu \in \reals^n$, satisfying the conditions of Theorem \ref{lnhdnlqwh} is called a \emph{Kaykobad pair}.
\end{definition}
We have the following result.
\begin{lemma}\label{knjxks}
Let $Q$ be an $n \times n$ diagonal matrix whose $i$-th diagonal entry is $\sign(q_i)$. Let $\bar B_J = B_J Q$ and $\bar q = Q q$.
If $(\bar B_J^\transp, \bar q)$ is a Kaykobad pair, then $(B_J, q)$ is a refining pair.
If $(-\bar B_J^\transp, \bar q)$ is a Kaykobad pair, then $(B_J, q)$ is a relaxing pair.
\end{lemma}
\begin{proof}
Suppose $(\bar B_J^\transp, \bar q)$ is a Kaykobad pair. Then $\bar B_J^\transp$ is invertible. We have $B_J (\bar B_J Q)^{-1} = B_J Q (\bar B_J )^{-1} = I$ and $(\bar B_J Q)^{-1} B_J = Q (\bar B_J)^{-1} (B_J Q) Q = I$, so $B_J$ is invertible. Moreover, we have
\[
0 < (\bar B_J^\transp)^{-1} \bar q = (Q B_J^\transp)^{-1} (Q q) = (B_J^\transp)^{-1} q,
\]
which means that $(B_J, q)$ is a refining pair.

If $(-\bar B_J^\transp, \bar q)$ is a Kaykobad pair, then $\bar B_J^\transp$ is invertible, which means that so is $B_J$. Moreover,
\[
0 < -(\bar B_J^\transp)^{-1} \bar q = -(Q B_J^\transp)^{-1} (Q q) = -(B_J^\transp)^{-1} q.
\]
Thus, $(B_J, q)$ is a relaxing pair.
\end{proof}

\algnewcommand{\LeftComment}[1]{\(\triangleright\) #1}

\begin{algorithm*}
\caption{Antecedents and consequents for linear inequality constraints by identifying systems of equations with positive solutions}\label{lnuqnw}
\hspace*{\algorithmicindent} \textbf{Input:} Term to transform $p^\transp x + q^\transp y \le r$, context $\Gamma$, \\
\hspace*{\algorithmicindent} \quad \quad \quad \;\;transform instruction $s$ (\textbf{true} for antecedents and \textbf{false} for consequents)\\
\hspace*{\algorithmicindent} \textbf{Output:} Transformed term $t'$ lacking any $y$ variables
\begin{algorithmic}[1]
\State MatrixRowTerms $\gets \emptyset$ \Comment{Rows of the context matrix $A$}
\State PartialSums $\gets \textproc{zeros}(\length(y))$
\State TCoeff $\gets -1$
\If{s}
\State TCoeff $\gets 1$
\EndIf
\For{$i = 1$ to $i = \length(y)$} \Comment{One iteration per row of context matrix}\label{hjqxg}
\State IthRowFound $\gets$ \textbf{false} \Comment{Indicate whether we could add the $i$-th row}
\For{$\gamma \in \Gamma \setminus \text{MatrixRowTerms}$}
\Statex \quad \quad \quad \LeftComment{1. Verifying Kaykobad pair: sign of nonzero matrix terms}
\State TermIsInvalid $\gets$ \textbf{false}
\For{$j = 1$ to $j = \length(y)$}
\If{$\coeff(\gamma, y_j) \ne 0$ and $\sign(\coeff(\gamma, y_j)) \ne \sign(q_j) \cdot \text{TCoeff}$}
\State TermIsInvalid $\gets$ \textbf{true}
\State \textbf{break}
\EndIf
\EndFor
\Statex \quad \quad \quad \LeftComment{2. Verifying Kaykobad pair: matrix diagonal terms}
\If{$\coeff(\gamma, y_i) = 0$ or TermIsInvalid}
\State \textbf{next}
\EndIf
\Statex \quad \quad \quad \LeftComment{3. Verifying Kaykobad pair: relationship between matrix and vector entries}
\State Residuals $\gets \text{zeros}(\length(y))$
\For{$j = 1$ to $j = \length(y)$}
\If{$j \ne i$}
\State Residuals$[j] \gets \sign(q_j) \cdot \text{TCoeff} \cdot \coeff(\gamma, y_j) \cdot \frac{q_i}{\coeff(\gamma, y_i)}$
\EndIf
\If{ $|q_j| \cdot \text{TCoeff} \le \text{PartialSums}[j] + \text{Residuals}[j]$}
\State TermIsInvalid $\gets$ \textbf{true}
\State \textbf{break}
\EndIf
\EndFor
\If{\textbf{not} TermIsInvalid}
\Statex \quad \quad \quad \quad \; \LeftComment{Resulting matrix is meeting Kaykobad pair conditions at $i$-th row}
\State IthRowFound $\gets$ \textbf{true}
\For{$j = 1$ to $j = \length(y)$}
\State $\text{PartialSums}[j] \gets \text{PartialSums}[j] + \text{Residuals}[j]$
\EndFor
\State MatrixRowTerms.append($\gamma$)
\State \textbf{break}
\EndIf
\EndFor
\If{\textbf{not} IthRowFound}
\State \Return Error: Cannot transform term
\EndIf
\EndFor\label{kdcja}
\State $B \gets \textproc{MatrixFromTerms}(\text{MatrixRowTerms}, y)$
\State $b \gets \textproc{VectorFromTerms}(\text{MatrixRowTerms}, y)$
\State \Return $p^\transp x + q^\transp B^{-1} b \le r$
\end{algorithmic}
\end{algorithm*}

Corollary \ref{lnhxql} and Lemma \ref{knjxks} yield a method for computing $Y$-antecedents and consequents of formulas in a context. To use it, we must construct $J$ such that $(B_J, q)$ meets the corollary's conditions. We construct $J$ by choosing $n$ formulas from the context $\Gamma$; these formulas must meet the conditions of a Kaykobad pair. One advantage of the Kaykobad condition is that it allows us to incrementally identify suitable constraints to add to the context $\Gamma'$, i.e., we don't have to select $n$ constraints before we run the verification. That is, when we have identified $k < n$ constraints, we can easily verify whether a candidate $(k+1)$-th formula would be acceptable for constructing a Kaykobad pair.
Algorithm \ref{lnuqnw} computes $Y$-antecedents and consequents for linear inequality constraints based on Corollary \ref{lnhxql}. Lines \ref{hjqxg}--\ref{kdcja} search the context $\Gamma$ for $n$ constraints meeting the Kaykobad conditions. The rest of the algorithm computes the $Y$-antecedents/consequents. 
If there are $n$ variables to be eliminated, and $N = |\Gamma|$ constraints in the context $\Gamma$, the algorithm has complexity $O(n^2 N + N^3)$. The function $\coeff(\gamma, y_j)$ extracts the coefficient of the variable $y_j$ from the term $\gamma$. The call $\textproc{MatrixFromTerms}(\text{MatrixRowTerms}, y)$ extracts all coefficients of the $y$ variables contained in MatrixRowTerms and makes these coefficients the rows of the resulting matrix.
The call $\textproc{VectorFromTerms}(\text{MatrixRowTerms}, y)$ returns a vector of all expressions contained in MatrixRowTerms with their $y$ variables removed. These are the elements of $b(x,z)$. Finally, $\diag(v)$ returns a diagonal matrix whose entries are the vector $v$.

We implemented Algorithm~\ref{lnuqnw} in Python and generated benchmarks for it.
We considered formulas $\phi$ of the form \eqref{jbhgbkgh} and contexts of the form \eqref{jvbhgkvgb}. We varied the number $N$ of constraints in the context, the number $m$ of variables to be eliminated, and the total number of variables present in the problem $n + m + o$ (in our experiments $\phi$ and $\Gamma$ had the same variables, so we had $o = 0$). We randomly generated coefficients for all variables in $\phi$ and $\Gamma$, which means that $\Gamma$ was always a dense matrix, a scenario unlikely to occur in applications.
We executed the algorithm for the situations in which we wanted to eliminate 2 irrelevant variables and 4 irrelevant variables. The results of the experiments are shown in Table~\ref{hqgjjncdc}.

\newcommand\headercell[1]{%
   \smash[b]{\begin{tabular}[t]{@{}c@{}} #1 \end{tabular}}}

\begin{table}[t]
\centering
\caption{Execution time in seconds of Algorithm~\ref{lnuqnw} for a given total number of variables, number of constraints in the context $\Gamma$, and number of irrelevant variables (or variables that need to be eliminated)}
\label{hqgjjncdc}
\begin{tabular}[t]{ccrrrrrr}
\toprule
        & \headercell{Constraints \\ in context} &    \multicolumn{6}{c}{Total number of variables} \\
        \cmidrule{3-8}
        & &    5 &   10 &   15 &   20 &   25 &   30 \\
\midrule
    &5 & 0.24 & 0.41 & 0.64 & 0.88 & 1.11 & 1.43 \\
2    &10 & 0.23 & 0.42 & 0.63 & 0.89 & 1.13 & 1.45 \\
irrelevant    &20 & 0.23 & 0.45 & 0.67 & 0.93 & 1.18 & 1.49 \\
variables    &100 & 0.39 & 0.7  & 1.06 & 1.44 & 1.88 & 2.47 \\
    &300 & 1.65 & 2.53 & 3.68 & 5.21 & 7.21 & 9.73 \\
\midrule
&   5 & 0.43 & 0.86 &   NA   & 1.74 &   NA   &   NA   \\
4&  10 & 0.45 & 0.86 &   1.31 & 1.8  &   2.21 &   2.65 \\
irrelevant&  20 & 0.46 & 0.88 &   1.34 & 1.88 &   2.29 &   2.7  \\
variables& 100 & 0.67 & 1.16 &   1.73 & 2.44 &   2.95 &   3.76 \\
& 300 & 2.05 & 3.1  &   4.5  & 6.36 &   8.23 &  10.95 \\
\bottomrule
\end{tabular}
\end{table}

\subsubsection{Computing J via linear programming}
Now we will build $J$ by numerically solving \eqref{lxkasjf} and \eqref{kdsgckg} for fixed values of $x$ and $z$.

\begin{lemma}
Let $a \in \reals^m$ and $c \in \reals^o$.
\begin{itemize}
\item Suppose $g^-(a, c)$ is finite and the optimum of the LP \eqref{lxkasjf} (with $x = a$ and $z = c$) is attained at $y^\star$. Let $J = \setArg{i}{b_i(a,c) - \sum_{j=1}^n \beta_{ij} y_j^\star = 0}$ and assume that $|J| = n$, where $n$ is the number of optimization variables $y$ in $g^-$.
If $B_J$ is invertible, then $(B_J, q)$ is a refining pair.
\item Similarly, suppose $g^+(a, c)$ is finite and the optimum of the LP \eqref{kdsgckg} (with $x = a$ and $z = c$) is attained at $y^\star$. Let $J = \setArg{i}{b_i(a,c) - \sum_{j=1}^n \beta_{ij} y_j^\star = 0}$ and assume that $|J| = n$.
If $B_J$ is invertible, then $(B_J, q)$ is a relaxing pair.
\end{itemize}
\end{lemma}
\begin{proof}
We prove the first part. Consider the following problems:
\begin{align*}
\text{primal}&
\begin{cases}
\begin{aligned}[t]
    & \underset{y}{\text{minimize}}
    & & - q^\transp y \\
    & \text{subject to}
    & &  B y \le b(a,c)
\end{aligned}
\end{cases}
\\
\text{dual}&
\begin{cases}
\begin{aligned}[t]
    & \underset{\lambda}{\text{maximize}}
    & & -b(a,c)^\transp \lambda \\
    & \text{subject to}
    & &  B^\transp \lambda - q= 0 \\
    & & & \lambda \ge 0
\end{aligned}
\end{cases}
\end{align*}
Since $g^-(a,c)$ is finite, the primal is feasible. By strong duality, so is the dual.
Let $\lambda^\star$ be the value of $\lambda$ where the dual attains its optimum. Then $\lambda^\star \ge 0$ and $0 = B^\transp \lambda ^\star - q = B_J^\transp \lambda_J^\star + B_{\hat J}^\transp \lambda_{\hat J}^\star - q$,
where $\hat J = \{1, \ldots, N\} \setminus J$.
Due to complementary slackness, we know that $\lambda_{\hat J}^\star = 0$. Thus, $0 = B_J^\transp \lambda_J^\star - q$. By assumption, $B_J$ is invertible. Then $(B_J, q)$ is a refining pair.
The proof of the second part is similar.
\end{proof}


\begin{algorithm}[h]
\caption{Antecedents and consequents for linear inequality constraints through linear programming}\label{kjdgck}
\hspace*{\algorithmicindent} \textbf{Input:} Term to transform $p^\transp x + q^\transp y \le r$, context $\Gamma$, \\
\hspace*{\algorithmicindent} \quad \quad \quad $a \in \reals^m$, $c \in \reals^o$, transform instruction $s$
\hspace*{\algorithmicindent} \quad \quad \quad (\textbf{true} for antecedents and \textbf{false} for consequents)\\
\hspace*{\algorithmicindent} \textbf{Output:} Transformed term $t'$ lacking any $y$ variables
\begin{algorithmic}[1]
\State $B \gets \textproc{MatrixFromTerms}(\Gamma, y)$
\State $b \gets \textproc{VectorFromTerms}(\Gamma, y)$
\State $b_e \gets \textproc{Evaluate}(b, a, c)$
\If{s}
\State $(\text{success}, y^\star) \gets \textproc{LinearProgramming}(-q, B, b_e)$
\Else
\State $(\text{success}, y^\star) \gets \textproc{LinearProgramming}(q, B, b_e)$
\EndIf
\If{\textbf{not} success}
\State \Return Error: LP 
is unfeasible
\EndIf
\State $S \gets b_e - B y^\star$
\State $J \gets \emptyset$
\For{$j = 1$ to $j = \length(b)$}
\If{$S_j = 0$}
\State $J \gets J \setunion \{j\}$
\EndIf
\EndFor
\State $(\text{success}, \hat B_J) \gets \textproc{MatrixInv}(B_J)$
\If{\textbf{not} success}
\State \Return Error: cannot invert $B_J$
\EndIf
\State \Return $p^\transp x + q^\transp \hat B_J b_J \le r$
\end{algorithmic}
\end{algorithm}

\medskip

Lemma \ref{kjdgck} allows us to obtain the solution to a linear programming problem with symbolic constraints $B y \le b(x,z)$ in a reduced context $B_J y \le b_J(x,z)$, where we identify $J$ by solving a numerical LP.
Lemma \ref{kjdgck} and Corollary \ref{lnhxql} yield a method for computing $Y$-antecedents and consequents. This method is reflected in Algorithm \ref{kjdgck}.
As before, $\textproc{MatrixFromTerms}(\Gamma, y)$ and $\textproc{VectorFromTerms}(\Gamma, \allowbreak y)$ extract from the context $\Gamma$ the matrix $B$ and symbolic vector $b(x,z)$ of the constraints $B y \le b(x,z)$. $\textproc{Evaluate}(b, a, c)$ returns the vector $b(a, c) \in \reals^N$. $\textproc{LinearProgramming}(q, \allowbreak B, b_e)$ solves the LP $\underset{y}{\text{min.}} \; q^\transp y$ subject to $B y \le b_e$ and returns a success variable and the value $y^\star$ where the minimum is attained. The success variable is true when the LP is feasible and has a finite solution. \textproc{MatrixInv} computes matrix inverses. Its success variable is false when the matrix is not invertible.

\section{Discussion and concluding remarks}


To the best of our knowledge, the identification of the problems of variable elimination via 
antecedent and consequent synthesis as relevant to the computation of specifications for requirement engineering
is new. We provided two efficient algorithms for the solutions of these problems. Both algorithms are sound but incomplete.

The synthesis problems we considered are closely related to quantifier elimination, of which there is a large body of work. In fact, the universal solutions to the synthesis problems are expressed as quantifications---see Proposition~\ref{kugbkgk}.

Quantifier-elimination algorithms for real arithmetic are given by Ferrante and Rackoff~\cite{ferrante1975decision}, 
Monniaux~\cite{monniaux2008quantifier}, Nipkow~\cite{10.1007/978-3-540-71070-7-3}, John and Chakraborty~\cite{john2016layered}, and others.
Audemard et al.~\cite{10.1007/3-540-45470-5-22} and
Bj{\o}rner~\cite{10.1007/978-3-642-14203-1-27} discuss linear quantifier elimination methods in the context of DPLL-based search.
Cousot and Halbwachs~\cite{10.1145/512760.512770} and Monniaux~\cite{monniaux2009automatic} apply it in the context of 
abstract interpretation.
The earliest means for carrying out quantifier elimination for linear inequalities is Fourier-Dines-Motzkin elimination~\cite{fourier,dines1919systems,motzkin1936beitrage}, to which many improvements have been made---see \cite{DANTZIG1973288,Duffin1974,lassez1992fourier,chandru1993variable,Imbert1993FouriersEW}.
All known algorithms have at least exponential worst case complexity, but can be extremely performant on many typical problems---see \cite{monniaux2008quantifier} for details.

An important use of the computation of specifications in requirement engineering is the support of design-space exploration and tradeoff analysis. When designing complex systems, engineers may need to navigate a large design space. To do this effectively, the computation of specifications has to be supported by efficient algorithms. This motivated us to look for efficient, though incomplete, algorithms for antecedent/consequent synthesis. Moreover, in requirement engineering, we want to produce outputs which are syntactically similar to the formulas from which variables are eliminated---see the discussion in Section~\ref{sc:pord}. This aspect motivated our results of Propositions~\ref{kjbhgbjk} and~\ref{kgbjsfgbj}.

Some next steps we perceive in contextual variable elimination for requirement engineering include better algorithms for polyhedral constraint synthesis and support for temporal logic. The algorithms we proposed to synthesize $Y$-antecedents and consequents are based on Corollary \ref{nlkflkg}, which yields optimal solutions. However, the algorithms we presented for solving \eqref{lxkasjf} and \eqref{kdsgckg} have room to improve.
Per Lemma~\ref{lnqfxq}, given a set of $N$ linear equations in $n$ variables ($N > n$),
it would be very useful to research methods to efficiently identify a set of $n$ linear equations with a nonnegative solution.
Finally, to support the computation of $Y$-antecedents and consequents of formulas in a context for temporal logic specifications, we believe the decomposition enabled by Proposition~\ref{kqxnfkxr} could have a useful role.

\bibliographystyle{ieeetr}
\bibliography{support/references}

\begin{thebibliography}{10}

\bibitem{BenvenisteContractBook}
A.~Benveniste, B.~Caillaud, D.~Nickovic, R.~Passerone, J.-B. Raclet,
  P.~Reinkemeier, A.~Sangiovanni-Vincentelli, W.~Damm, T.~A. Henzinger, and
  K.~G. Larsen, ``Contracts for system design,'' {\em Foundations and
  Trends$^{\text{\scriptsize{\textregistered}}}$\hspace{-.3em} in Electronic
  Design Automation}, vol.~12, no.~2-3, pp.~124--400, 2018.

\bibitem{Incer:EECS-2022-99}
I.~Incer, {\em The Algebra of Contracts}.
\newblock PhD thesis, EECS Department, University of California, Berkeley, May
  2022.

\bibitem{DBLP:journals/ejcon/Sangiovanni-VincentelliDP12}
A.~L. Sangiovanni{-}Vincentelli, W.~Damm, and R.~Passerone, ``Taming {D}r.
  {F}rankenstein: Contract-based design for cyber-physical systems,'' {\em Eur.
  J. Control}, vol.~18, no.~3, pp.~217--238, 2012.

\bibitem{boyd2004convex}
S.~Boyd and L.~Vandenberghe, {\em Convex Optimization}.
\newblock Cambridge university press, 2004.

\bibitem{KAYKOBAD1985133}
M.~Kaykobad, ``Positive solutions of positive linear systems,'' {\em Linear
  Algebra and its Applications}, vol.~64, pp.~133--140, 1985.

\bibitem{ferrante1975decision}
J.~Ferrante and C.~Rackoff, ``A decision procedure for the first order theory
  of real addition with order,'' {\em SIAM Journal on Computing}, vol.~4,
  no.~1, pp.~69--76, 1975.

\bibitem{monniaux2008quantifier}
D.~Monniaux, ``A quantifier elimination algorithm for linear real arithmetic,''
  in {\em Logic for Programming, Artificial Intelligence, and Reasoning: 15th
  International Conference, LPAR 2008, Doha, Qatar, November 22-27, 2008.
  Proceedings 15}, pp.~243--257, Springer, 2008.

\bibitem{10.1007/978-3-540-71070-7-3}
T.~Nipkow, ``Linear quantifier elimination,'' in {\em Automated Reasoning}
  (A.~Armando, P.~Baumgartner, and G.~Dowek, eds.), (Berlin, Heidelberg),
  pp.~18--33, Springer Berlin Heidelberg, 2008.

\bibitem{john2016layered}
A.~K. John and S.~Chakraborty, ``A layered algorithm for quantifier elimination
  from linear modular constraints,'' {\em Formal Methods in System Design},
  vol.~49, pp.~272--323, 2016.

\bibitem{10.1007/3-540-45470-5-22}
G.~Audemard, P.~Bertoli, A.~Cimatti, A.~Korni{\l}owicz, and R.~Sebastiani,
  ``Integrating boolean and mathematical solving: Foundations, basic
  algorithms, and requirements,'' in {\em Artificial Intelligence, Automated
  Reasoning, and Symbolic Computation} (J.~Calmet, B.~Benhamou, O.~Caprotti,
  L.~Henocque, and V.~Sorge, eds.), (Berlin, Heidelberg), pp.~231--245,
  Springer Berlin Heidelberg, 2002.

\bibitem{10.1007/978-3-642-14203-1-27}
N.~Bj{\o}rner, ``Linear quantifier elimination as an abstract decision
  procedure,'' in {\em Automated Reasoning} (J.~Giesl and R.~H{\"a}hnle, eds.),
  (Berlin, Heidelberg), pp.~316--330, Springer Berlin Heidelberg, 2010.

\bibitem{10.1145/512760.512770}
P.~Cousot and N.~Halbwachs, ``Automatic discovery of linear restraints among
  variables of a program,'' in {\em Proceedings of the 5th ACM SIGACT-SIGPLAN
  Symposium on Principles of Programming Languages}, POPL '78, (New York, NY,
  USA), p.~84–96, Association for Computing Machinery, 1978.

\bibitem{monniaux2009automatic}
D.~P. Monniaux, ``Automatic modular abstractions for linear constraints,'' {\em
  ACM SIGPLAN Notices}, vol.~44, no.~1, pp.~140--151, 2009.

\bibitem{fourier}
J.~Fourier, ``Solution d'une question particulière du calcul des
  inégalités,'' {\em Nouveau Bulletin des sciences par la Société
  philomathique de Paris, p. 99}, pp.~317--319, 1826.

\bibitem{dines1919systems}
L.~L. Dines, ``Systems of linear inequalities,'' {\em Annals of Mathematics},
  pp.~191--199, 1919.

\bibitem{motzkin1936beitrage}
T.~S. Motzkin, {\em Beitr{\"a}ge zur Theorie der linearen Ungleichungen}.
\newblock PhD thesis, University of Basel, 1936.

\bibitem{DANTZIG1973288}
G.~B. Dantzig and B.~{Curtis Eaves}, ``Fourier-motzkin elimination and its
  dual,'' {\em Journal of Combinatorial Theory, Series A}, vol.~14, no.~3,
  pp.~288--297, 1973.

\bibitem{Duffin1974}
R.~J. Duffin, ``On {F}ourier's analysis of linear inequality systems,'' in {\em
  Pivoting and Extension: In honor of A.W. Tucker} (M.~L. Balinski, ed.),
  pp.~71--95, Berlin, Heidelberg: Springer Berlin Heidelberg, 1974.

\bibitem{lassez1992fourier}
J.-L. Lassez and M.~J. Maher, ``On {F}ourier's algorithm for linear arithmetic
  constraints,'' {\em Journal of Automated Reasoning}, vol.~9, pp.~373--379,
  1992.

\bibitem{chandru1993variable}
V.~Chandru, ``Variable elimination in linear constraints,'' {\em The Computer
  Journal}, vol.~36, no.~5, pp.~463--472, 1993.

\bibitem{Imbert1993FouriersEW}
J.~Imbert, ``Fourier's elimination: Which to choose?,'' in {\em Principles and
  Practice of Constraint Programming}, pp.~117--129, 1993.

\end{thebibliography}



\end{document}